\documentclass[12pt,a4paper]{article}
\usepackage{amsmath,amsfonts,amssymb}
\usepackage{fullpage}

\newcommand{\AC}{\mathcal{AC}}
\newcommand{\floor}[1]{{\lfloor#1\rfloor}}
\newcommand{\ceil}[1]{{\lceil#1\rceil}}
\newcommand{\diam}{{\rm diam}}
\newcommand{\dist}{{\rm dist}}
\renewcommand{\span}{{\rm span}}
\newcommand{\eq}{{\rm eq}}

\newcommand{\NP}{\ensuremath{\mathcal{NP}}}
\newcommand{\coNP}{\ensuremath{{\rm co}\mathcal{NP}}}
\renewcommand{\P}{\ensuremath{\mathcal{P}}}
\newcommand{\N}{{\mathbb N}}

\newcommand{\eps}{\varepsilon}
\newcommand{\seq}{\subseteq}

\newcommand{\qed}{\rule{7pt}{7pt}}

\newtheorem{theorem}{Theorem}[section]
\newtheorem{corollary}[theorem]{Corollary}
\newtheorem{proposition}[theorem]{Proposition}
\newtheorem{definition}[theorem]{Definition}
\newtheorem{lemma}[theorem]{Lemma}
\newtheorem{claim}[theorem]{Claim}

\newtheorem{conjecture}[theorem]{Conjecture}

\newtheorem{problem}[theorem]{Problem}

\newenvironment{proof}{\noindent{\bf Proof}\hspace*{1em}}{\qed\bigskip}
\newenvironment{proof-sketch}{\noindent{\bf Sketch of Proof}\hspace*{1em}}{\qed\bigskip}
\newenvironment{proof-idea}{\noindent{\bf Proof Idea}\hspace*{1em}}{\qed\bigskip}
\newenvironment{proof-of-lemma}[1]{\noindent{\bf Proof of Lemma #1}\hspace*{1em}}{\qed\bigskip}
\newenvironment{proof-of-claim}[1]{\noindent{\bf Proof of Claim #1}\hspace*{1em}}{\qed\bigskip}
\newenvironment{proof-of-thm}[1]{\noindent{\bf Proof of Theorem #1}\hspace*{1em}}{\qed\bigskip}
\newenvironment{proof-attempt}{\noindent{\bf Proof Attempt}\hspace*{1em}}{\qed\bigskip}
\newenvironment{proofof}[1]{\noindent{\bf Proof of #1:}\hspace*{1em}}{\qed\bigskip}
\newenvironment{remark}{\noindent{\bf Remark}\hspace*{1em}}{\bigskip}

\title{Acquaintance Time of a Graph}

\newcommand*\samethanks[1][\value{footnote}]{\footnotemark[#1]}
\author{
  Itai Benjamini\thanks{
    Faculty of Mathematics and Computer Science,
    Weizmann Institute of Science, Rehovot, {\sc Israel}.
    Emails:{\tt\{itai.benjamini,igor.shinkar,gilad.tsur\}@weizmann.ac.il}.
    Igor Shinkar's research is supported by ERC grant number 239985.
    }
  \and
  Igor Shinkar\samethanks
  \and
  Gilad Tsur\samethanks
}

\begin{document}

\maketitle
\thispagestyle{empty}

\begin{abstract}

We define the following parameter of connected graphs.
For a given graph $G = (V,E)$ we place one agent in each vertex $v \in V$.
Every pair of agents sharing a common edge is declared to be acquainted.
In each round we choose some matching of $G$ (not necessarily a maximal matching),
and for each edge in the matching the agents on this edge swap places.
After the swap, again, every pair of agents sharing a common edge become acquainted,
and the process continues. We define the \emph{acquaintance time}
of a graph $G$, denoted by $\AC(G)$, to be the minimal number of rounds required
until every two agents are acquainted.

We first study the acquaintance time for some natural families of graphs
including the path, expanders, the binary tree, and the complete bipartite
graph. We also show that for all $n \in \N$
and for all positive integers $k \leq n^{1.5}$
there exists an $n$-vertex graph $G$ such that
$k/c \leq \AC(G) \leq c \cdot k$ for some universal constant $c \geq 1$.
We also prove that for all $n$-vertex connected graphs $G$ we have
$\AC(G) = O\left(\frac{n^2}{\log(n)/\log\log(n)}\right)$,
thus improving the trivial upper bound of $O(n^2)$
achieved by sequentially letting each agent perform depth-first search
along some spanning tree of $G$.

Studying the computational complexity of this problem, we prove that for any
constant $t \geq 1$ the problem of deciding that a given graph $G$ has
$\AC(G) \leq t$ or $\AC(G) \geq 2t$ is $\NP$-complete. That is, $\AC(G)$ is
$\NP$-hard to approximate within multiplicative factor of 2, as well as within
any additive constant factor.

On the algorithmic side, we give a deterministic algorithm that given
an $n$-vertex graph $G$ with $\AC(G)=1$ finds a strategy for
acquaintance that consists of $\ceil{n/c}$ matchings in time $n^{c+O(1)}$.
We also design a randomized polynomial time algorithm that given
an $n$-vertex graph $G$ with $\AC(G)=1$ finds with high probability
an $O(\log(n))$-rounds strategy for acquaintance.

\end{abstract}

%
\newpage
\pagenumbering{arabic}
\section{Introduction}\label{sec:intro}

In this work we deal with the following problem: agents walk on a graph meeting each
other, and our goal is to make every pair of agents meet as fast as possible.
Specifically, we introduce the following parameter of connected graphs.
For a given graph $G = (V,E)$ we place one agent in each vertex of the graph.
Every pair of agents sharing a common edge is declared to be acquainted.
In each round we choose some matching of $G$ (not necessarily a maximal matching),
and for each edge in the matching the agents on this edge swap places.
After the swap, again, every pair of agents sharing a common edge become acquainted,
and the process continues until every two agents are acquainted with each other.
Such a sequence is called {\em a strategy for acquaintance in $G$}.
We define the \emph{acquaintance time}
of a graph $G$, denoted by $\AC(G)$, to be the minimal number
of rounds in a strategy for acquaintance in $G$.

In order to get some feeling regarding this parameter
note that if for a given graph $G$ a list of matchings $(M_1,\dots,M_t)$
is a witness-strategy for the assertion that $\AC(G) \leq t$,
then the inverse list $(M_t,\dots,M_1)$ is also a witness-strategy for this assertion.
We remark that in general a witness-strategy is not commutative in the order of the matchings.%
\footnote{For example, let $G=(V=\{1,2,3,4\}, E=\{(1,2),(2,3),(3,4)\})$ be
the path of length 4. Then, the sequence
$(M_1=\{(1,2)\},M_2=\{(3,4)\},M_3=\{(1,2)\})$ is a strategy for acquaintance
in $G$, whereas, the sequence
$(M_1=\{(1,2)\},M_3=\{(1,2)\},M_2=\{(3,4)\})$ is not.}
For a trivial bound of $\AC(G)$ we have $\AC(G) \geq \floor{\diam(G)/2}$,
where $\diam(G)$ is the maximal distance between two vertices of $G$.
It is also easy to see that for every graph $G = (V,E)$ with $n$ vertices
it holds that $\AC(G) \geq \frac{{n \choose 2}}{|E|} - 1$.
Indeed, before the first round exactly $|E|$ pairs of agents are acquainted.
Similarly, in each round at most $|E|$ new pairs get acquainted.
This implies that $|E| + \AC(G) \cdot |E| \geq {n \choose 2}$,
since in any solution the total number of pairs that met up
to time $\AC(G)$ is ${n \choose 2}$.
For an upper bound, for every graph $G$ with $n$ vertices we have
$\AC(G) \leq 2n^2$, as every agent can meet all others by
traversing the graph along some spanning tree in at most $2n$ rounds.

Note that for $t \in \N$ the problem of deciding whether a graph $G$
has $\AC(G) \leq t$ is in $\NP$, and the natural $\NP$-witness is
a strategy for acquaintance in $G$.
This problem is different from many classical
$\NP$-complete problems, such as graph coloring or vertex cover,
in the sense that checking an $\NP$-witness for $\AC(G)$
is ``dynamic'', and involves evolution in time.
Several problems of similar
flavor have been studied in the past, including
the well studied problems of
Gossiping and Broadcasting (see the survey of
Hedetniemi, Hedetniemi, and Liestman~\cite{HHL} for details),
Collision-Free Network Exploration (see~\cite{CDGKKP}),
and the Target Set Selection Problem (see, e.g., \cite{KKT,Ch,Rei}).
One such problem of particular relevance is Routing
Permutation on Graphs via Matchings studied by Alon, Chung, and
Graham in~\cite{ACG}. In this problem the input is a graph $G=(V,E)$ and a
permutation of the vertices $\sigma: V\to V$, and the goal is to route
all agents to their respective destinations according to $\sigma$;
that is, the agent sitting originally in the vertex $v$ should be
routed to the vertex $\sigma(v)$ for all $v \in V$.
In our setting we encounter a similar routing problem, where we route
the agents from some set of vertices $S \seq V$ to some $T \seq V$
without specifying the target location in $T$ of each of the agents.


%
%
\subsection{Our results}\label{sec:our-section}

We start this work by providing asymptotic computations of the
acquaintance time for some interesting families of graphs.
For instance, if $P_n$ is the path of length $n$, then $\AC(P_n) = O(n)$,
which is tight up to a multiplicative constant, since $\diam(P_n) = n-1$.
In particular, this implies that $\AC(H) = O(n)$
for all Hamiltonian graphs $H$ with $n$ vertices.
We also prove that for constant degree expanders $G = (V,E)$
on $n$ vertices the acquaintance time is $O(n)$, which is tight,
as $|E|= O(n)$ and $\AC(G) = \Omega(\frac{n^2}{|E|})$.
More examples include the binary tree, the complete bipartite graph, and the barbell graph.

We then provide examples of graphs with different ranges of the acquaintance time.
We show in Theorem~\ref{thm:range of AC} that for all $n \in \N$
and for all positive integers $k \leq n^{1.5}$
there exists an $n$-vertex graph $G$ such that
$k/c \leq \AC(G) \leq c \cdot k$ for some universal constant $c \geq 1$.
Another interesting result says that for every connected graph $G$ with $n$ vertices
the acquaintance time is, in fact, asymptotically smaller than the trivial $O(n^2)$ bound.
Specifically, we prove in Theorem~\ref{thm:AC-bound} that for every connected graph $G$
with $n$ vertices $\AC(G) \leq \frac{c n^2}{\log(n)/\log\log(n)}$
for some absolute constant $c$.

We also study the problem of computing/approximating $\AC(G)$ for a given graph $G$.
As noted above, for $t \in \N$ the problem of deciding whether a given graph $G$
has $\AC(G) \leq t$ is in $\NP$, and the natural $\NP$-witness is a sequence
of $t$ matchings that allows every two agents to get acquainted.
We prove that the acquaintance time problem is $\NP$-complete,
by showing a reduction from the graph coloring problem.
Specifically, Theorem~\ref{thm:NP-hardness} says that
for every $t \geq 1$ it is $\NP$-hard to distinguish whether
a given graph $G$ has $\AC(G) \leq t$ or $\AC(G) \geq 2t$.
Hence, $\AC(G)$ is $\NP$-hard to approximate within a
multiplicative factor of 2, as well as within any additive constant.
In fact, we conjecture that it is $\NP$-hard to approximate $\AC$
within any multiplicative factor.

On the algorithmic side we study graphs whose acquaintance time equals to 1.
We show there is a deterministic algorithm that when
given an $n$-vertex graph $G$ with $\AC(G)=1$ finds an
$\ceil{n/c}$-rounds strategy for acquaintance in $G$ in time $n^{c+O(1)}$.
We also design a randomized polynomial time algorithm that when given
an $n$-vertex graph $G$ with $\AC(G)=1$ finds with high probability
an $O(\log(n))$-rounds strategy for acquaintance.

\subsection{Recent Developments}\label{sec:recent dev}

Recently, based on a preprint of our results posted on the internet~\cite{BST},
there have been several developments on this problem.

    Kinnersley~et~al.~\cite{KMP} proved that for all graphs $G$ with $n$
    vertices it holds that $\AC(G) = O(n^2 / \log(n))$, which improves the
    $O\left(\frac{n^2}{\log(n)/\log\log(n)}\right)$ bound in this paper.
    More recently, Angel and Shinkar~\cite{AS} proved that for all graphs $G$ with $n$
    vertices it holds that $\AC(G) = O(\Delta n)$, where $\Delta$ is the maximal degree
    of $G$. By combining this result with the bound $\AC(G) = O(n^2/ \Delta)$
    shown in Claim~\ref{claim:AC-bound maxdeg} in this
    paper it follows that for every graph $G$ with $n$
    vertices it holds that $\AC(G) = O(n^{1.5})$. This upper bound is tight up
    to a multiplicative constant as shown in Theorem~\ref{thm:range of AC}.

    Kinnersley~et~al.~\cite{KMP} also studied $\AC$ for random graphs, and
    proved that for $G(n,p)$ with $p > \frac{(1+\eps)\ln(n)}{n}$ for some $\eps>0$
    (slightly above the threshold for connectivity) it holds that
    $\AC(G(n,p)) = O(\frac{\log(n)}{p})$ with high probability.
    This result is tight up to an $O(\log(n))$ multiplicative factor since
    $\AC(G(n,p)) \geq \Omega(\frac{1}{p})$ by the trivial bound on the
    number of edges.
    In fact, it is shown in~\cite{KMP} that for $p > n^{-1/2+\eps}$ the bound
    $\AC(G(n,p)) = O(\frac{\log(n)}{p})$
    is tight up to a constant multiplicative factor.

M\"uller and Pra{\l}at~\cite{MP} studied the acquaintance time of a random subgraph
of a random geometric graph, where $n$ vertices are chosen independently
uniformly at random from $[0,1]^2$, and two vertices are adjacent with probability $p$
if the Euclidean distance between them is at most $r$. They show asymptotic results
for $\AC(G(n,r,p))$ for a wide range of $r=r(n)$ and $p=p(n)$.

\section{Definitions and Notation}\label{sec:definitions}

Throughout the paper all graphs are simple and undirected.
We use standard notations for the standard parameters of graphs.
Given a graph $G=(V,E)$ and two vertices $u,v \in V$
the distance between $u$ and $v$, denoted by
$\dist(u,v)$, is the length of a shortest path from $u$ to $v$ in $G$.
For a vertex $v$ and a set of vertices $U \seq V$ the distance of $v$ from $U$
is defined to be $\dist(v,U) = \min_{u \in U}\dist(v,u)$.
The diameter of the graph $G$, denoted by $\diam(G)$, is the maximal
distance between two vertices of the graph.
For a vertex $u \in V$ the set of neighbors of $u$ is denoted by $N(u) = \{ w \in V : (u,w) \in E\}$.
Similarly, for a set $U \seq V$ the set of neighbors of $U$
is $N(U) = \{ w \in V : (u,w) \in E \mbox{ for some } u \in U \}$.
The {\em independence number of $G$}, denoted by $\alpha(G)$,
is the cardinality of the largest independent set, that is,
a set of vertices in the graph, no two of which are adjacent.
The {\em chromatic number of $G$}, denoted by $\chi(G)$,
is the minimal number $c \in \N$ such that there is a coloring $f : V \to [c]$
of the vertices that satisfies $f(v) \neq f(w)$ for all edges $(v,w) \in E$.
The {\em equi-chromatic number of $G$}, denoted by $\chi_\eq(G)$, is the
minimal number $c \in \N$ such that there is a {\em balanced}
coloring $f : V \to [c]$ that satisfies $f(v) \neq f(w)$ for all edges $(v,w) \in E$,
where a coloring $f : V \to [c]$ is said to be balanced
if $|f^{-1}(i)| = |f^{-1}(j)|$ for all $i,j \in [c]$.

For a given graph $G = (V,E)$ the acquaintance time is defined as follows.
We place one agent in each vertex $v \in V$.
Every pair of agents sharing a common edge is declared to be acquainted.
In each round we choose a matching of $G$,
and for each edge in the matching the agents on this edge swap places.
After the swap, again, every pair of agents sharing a common edge become
acquainted, and the process continues.
A sequence of matchings in the graph is called {\em a strategy}.
A strategy that allows every pair of agents
to meet is called {\em a strategy for acquaintance in $G$}.
The acquaintance time of $G$, denoted by $\AC(G)$,
is the minimal number of rounds required for such a strategy.

As mentioned in the introduction, this problem is related to a certain routing
problem studied in~\cite{ACG}. Specifically, we are
interested in the routing task summarized in the following claim.
For a given tree $G = (V,E)$ the claim gives a strategy for fast
routing of the agents from some set
of vertices $S \seq V$ to $T \seq V$ without specifying the target
location in $T$ of each of the agents.
\begin{claim}\label{claim:routing subgraphs}
    Let $G = (V,E)$ be a tree. Let $S,T \seq V$ be two
    subsets of the vertices of equal size $k = |S| = |T|$,
    and let $\ell = \max_{v \in S,u \in T}\{ \dist(v,u)  \}$
    be the maximal distance between a vertex in $S$ and a vertex
    in $T$. Then, there is a strategy of $\ell + 2(k-1)$ matchings
    that routes all agents from $S$ to $T$.
\end{claim}

\begin{proof}
    Let $G=(V,E)$ be a tree, and let $S,T \seq V$ be two subsets of the vertices
    of $G$. The proof is by induction on $k$.
    For the case of $k=1$ the statement is trivial,
    as $\ell$ rounds are enough to route a single agent.

    For the induction step let $k \geq 2$, and assume for simplicity
    that the only agents in the graph are those sitting in $S$,
    and our goal is to route them to $T$.
    Let $\span(S)$ be the minimal subtree of $G$
    containing all vertices $s \in S$, and define $\span(T)$ analogously.

    Note first that if $\span(S)=\span(T)$, then by minimality of the span
    any leaf $s^*$ of $\span(S)$ belongs to both $S$ and $T$.
    Therefore, we may apply the induction hypothesis to route the agents
    from $S \setminus \{s^*\}$ to $T \setminus \{s^*\}$ in the subtree
    $\span(S) \setminus \{s^*\}$, leaving the agent from $s^*$ in place,
    which proves the induction step for the case of $\span(S)=\span(T)$.

    Otherwise, let us assume without loss of generality that there is some $s^* \in S$
    that is not contained in $\span(T)$. If not, then we can consider the problem
    of routing the agents from $T$ to $S$, and note that viewing this strategy
    in the reverse order produces a strategy for routing from $S$ to $T$.

    Let $t^* \in T$ be a vertex such that $\dist(s^*,t^*) = \dist(s^*,T)$,
    and let $P = (s^*=p_0,p_1,\dots,p_r,p_{r+1}=t^*)$ be the shortest path from $s^*$ to $t^*$
    in $G$ (note that $r \leq \ell$ by definition of $\ell$).
    By the induction hypothesis, there is a strategy consisting of $\ell + 2(k-2)$
    rounds that routes the agents from $S \setminus \{s^*\}$ to $T \setminus \{t^*\}$.
    In such a strategy after the last step all agents are in $T \setminus \{t^*\}$
    and thus the vertices $\{p_1,\dots,p_r\}$ contain no agents
    (since $p_i \notin T\setminus\{t^*\}$ for all $i \in [r]$).
    After round number $(\ell + 2(k-2) - 1)$, i.e., one step before the last,
    the vertices $\{p_1,\dots,p_{r-1}\}$ contain no agents, because $\dist(p_i,T) \geq 2$
    for all $i \leq r-1$. Analogously, for all $j \leq r$ the vertices $\{p_1,\dots,p_{r-j}\}$
    contain no agent after round number $(\ell + 2(k-2) - j)$.
    Therefore, we can augment the strategy by moving the agent from $s^*$
    to $t^*$ along the path $P$. Specifically, for all $i=0,\dots,r$ we
    move the agent from $p_i$ to $p_{i+1}$ in round $\ell+2(k-2)-r+i+2$,
    which adds two rounds to the strategy. The claim follows.
\end{proof}

\section{Some Concrete Examples}\label{sec:examples}

We start with an easy example, showing that for the graph $P_n$,
a path of length $n$, the acquaintance time is $\Theta(n)$.

\begin{proposition}{\em($\AC$ of a path):}\label{prop:path}
    Let $P_n$ be a path of length $n$. Then $\AC(P_n) = \Theta(n)$.
\end{proposition}
\begin{proof}
    Clearly $\AC(P_n) \geq \floor{\diam(P_n)/2} = \floor{(n-1)/2}$.
    For the upper bound denote the vertices of $P_n$ by $v_1,\dots,v_n$,
    where $v_i$ is connected to $v_{i+1}$ for all $i \in [n-1]$,
    and denote by $p_i$ the agent sitting initially in the vertex $v_i$.
    Consider the following strategy that works in $O(n)$ rounds:
    \begin{enumerate}
    \item
        Apply Claim~\ref{claim:routing subgraphs} in order
        to route all agents $p_1,\dots,p_{\floor{n/2}}$ to the vertices $v_{\ceil{n/2}+1},\dots,v_n$,
        and route $p_{\floor{n/2}+1},\dots,p_n$ to the vertices $v_1,\dots,v_{\ceil{n/2}}$.
        This can be done in $O(n)$ rounds.
        Note that after this sequence every pair of agents $(p_i,p_j)$
        with $1 \leq i \leq \floor{n/2} < j \leq n$ have already met each other.
    \item
        Repeat the above procedure recursively on each of the two halves
        $(v_1,\dots,v_{\ceil{n/2}})$ and $(v_{\ceil{n/2}+1},\dots,v_n)$
        simultaneously.
    \end{enumerate}
    To bound the total time $T(n)$ of the procedure,
    we make $O(n)$ rounds in the first part, and at most $T(\ceil{n/2})$
    in the remaining parts. This gives us a bound of
    \[
        T(n) = O(n) + T(\ceil{n/2}) = O(n),
    \]
    as required.
\end{proof}

The following corollary is immediate from Proposition~\ref{prop:path}.
\begin{corollary}\label{cor:hamiltonian}
    Let $G$ be Hamiltonian graph with $n$ vertices.
    Then $\AC(G) = O(n)$.
\end{corollary}

We next prove that for constant degree expanders
the acquaintance time is also linear in the size of the graph.
For $\alpha > 0$ a $d$-regular graph $G=(V,E)$ with $n$ vertices is said to be an $(n,d,\alpha)$-expander
if for every subset $S \seq V$ of size $|S| \leq |V|/2$ it holds that $|N(S) \setminus S| \geq \alpha \cdot |S|$.
\begin{proposition}{\em($\AC$ of expander graphs):}\label{prop:expander}
    Let $G=(V,E)$ be an $(n,d,\alpha)$-expander graph for some $\alpha > 0$.
    Then $\AC(G) = \Theta(n)$, where the multiplicative constant
    in the $\Theta()$ notation depends only on $\alpha$ and $d$ but not on $n$.
\end{proposition}

\begin{proof}
    Recall that $\AC(G) = \Omega(\frac{n^2}{|E|}) = \Omega(n)$,
    since the expander is of constant degree.
    For the upper bound we shall need the following theorem due to
    Bj\"{o}rklund, Husfeldt and Khanna~\cite{BHK}
    saying that every expander graph contains a simple path of linear length.
    \begin{theorem}[{\cite[Theorem 4]{BHK}}]\label{thm:expander long paths}
        Let $G$ be an $(n,d,\alpha)$-expander graph.
        Then, $G$ contains a simple path of length $\Omega(\frac{\alpha}{d} \cdot n)$.
    \end{theorem}

    Let $P$ be a simple path of even length $\ell$ in $G$, where $\ell = \Omega(n)$.
    Such a path exists by Theorem~\ref{thm:expander long paths}.
    Partition all agents into $c = \ceil{2n/\ell}$
    disjoint subsets $C_1 \cup \dots  \cup C_c$ each of size at most $\ell/2$.
    Then, for every pair $i,j \in [c]$ we use the strategy from Claim~\ref{claim:routing subgraphs}
    to place the agents from the two subsets $C_i \cup C_j$ on $P$,
    and then apply the strategy from Proposition~\ref{prop:path}
    so that every pair of agents from $C_i \cup C_j$ meet.
    By repeating this strategy for every $i,j \in [c]$, we make sure
    that every pair of agents on $G$ meet each other.
    In order to analyze the total length of the strategy,
    we note that for a single pair $i,j \in [c]$
    the total time is at most $n + O(\ell)$,
    and hence the total length of the strategy
    is at most
    \[
        \AC(G) \leq {c \choose 2} \cdot O( n + \ell ),
    \]
    which is linear in $n$, since $\ell = \Omega(n)$, and $c = O(n/\ell) = O(1)$.
\end{proof}

Next we upper bound the acquaintance time of the binary tree graph.
\begin{proposition}{\em($\AC$ of binary tree):}\label{prop:binary tree}
    Let $T$ be the binary tree with $n$ vertices.
    Then $\AC(T) = O(n \log(n))$.
\end{proposition}

Note that $\AC(T) = \Omega(n)$ since the number of edges
in $T$ is $n-1$. A recent result of~\cite{AS} gives an
asymptotically tight bound of $\AC(T) = O(n)$,
where $O()$ hides some absolute constant.

\medskip

\begin{proof}
    Associate the vertices of $T$ with $\{0,1\}^{\leq \log(n)}$
    the binary strings of length at most $\log(n)$ in the following natural way.
    The vertex $v_\eps$ (corresponding to the empty string) is the root,
    the vertices $v_0$ and $v_1$ are the children of $v_\eps$.
    In general, the children of $v_{s}$ are $v_{s0}$ and $v_{s1}$.

    For $s \in \{0,1\}^{\leq \log(n)}$ denote by $T_s$ the subtree
    rooted at $v_s$.
    We also denote by $p_s$ the agent originally located in $v_s$,
    and let $P_s$ be the set of agents who were originally in $T_s$.

    We claim that it is enough to find a strategy of length $O(n \log(n))$
    that allows every agent in $P_0$ meet every agent from $P_1$.
    Indeed, suppose we have such strategy.
    We describe a strategy for acquaintance in $T$.
    \begin{enumerate}
    \item
        Let the agent $p_\eps$ sitting in $v_\eps$ meet all other agents by performing
        a DFS walk on the tree, and return everyone to their original locations by applying
        the same strategy in the reverse order.
        This step can be done in $O(n)$ rounds.
    \item
        Apply a strategy of length $O(n \log(n))$ that makes all agents in
        $P_0$ meet all agents in $P_1$.
    \item
        Return the agents of $P_0$ to $T_0$ and return the agents of $P_1$ to $T_1$,
        by reversing the order of steps taken in the previous item.
    \item
        Apply steps 1-3 recursively on the subtree $T_0$
        and on the subtree $T_1$ simultaneously.
        That is, every agent from $P_{00}$ meets every agent from $P_{01}$,
        and every agent from $P_{10}$ meets every agent from $P_{11}$,
        and so on...
    \end{enumerate}
    Analyzing the total number of rounds, we have $O(n \log(n))$ rounds in the
    first 3 steps. Therefore, the total number of rounds is upper bounded by
    $O(n \log(n)) + O(n/2 \log(n/2)) + O(n/4 \log(n/4)) \dots = O(n \log(n))$,
    as required.

\medskip
    Next, we describe a strategy that makes every agent from $T_0$
    meet every agent from $T_1$ in $O(n \log(n))$ rounds.
    \begin{enumerate}
    \item
        Let the agents $p_0$ and $p_1$ meet all other agents,
        and ignore them from now on.
        This step can be done in $O(n)$ rounds.
    \item
        Route the agents in $P_{00}$ to the subtree $T_{10}$, and route the
        agents in $P_{10}$ to $T_{00}$. This can be done in $O(n)$ rounds
        by considering the subtree of $T$ induced by the vertices
        $T_{00} \cup T_{10} \cup \{v_\eps,v_0,v_1\}$
        and applying Claim~\ref{claim:routing subgraphs}.
    \item
        Apply induction on the depth of the tree to make all agents in $P_{00}$
        (who are located in $T_{10}$) meet all agents in $P_{11}$
        (located in $T_{11}$), and simultaneously make all agents in
        $P_{01}$ (who are located in $T_{00}$) meet all agents in $P_{01}$
        (located in $T_{01}$).
    \item
        Route the agents in $P_{01}$ to the subtree $T_{10}$,
        route the agents in $P_{11}$ to $T_{11}$,
        route the agents in $P_{00}$ to $T_{00}$,
        and route the agents in $P_{10}$ to $T_{01}$.
        This can be done in $O(n)$ rounds by applying
        Claim~\ref{claim:routing subgraphs} on the appropriate subgraphs.
    \item
        Apply induction on the depth of the tree to make all agents in $P_{01}$
        (who are located in $T_{10}$) meet all agents in $P_{11}$
        (located in $T_{11}$), and simultaneously make all agents in
        $P_{00}$ (who are located in $T_{00}$) meet all agents in $P_{10}$
        (located in $T_{01}$).
    \end{enumerate}
    It is clear that in steps 1,3, and 5 all agents from $P_0$ meet all agents
    from $P_1$. For the analysis of the number of rounds let us denote the
    total number of rounds by $T(n)$. Then, steps 1,2, and 4 contribute $O(n)$ rounds
    to $T(n)$, and steps 3 and 5 contribute additional $2T(n/2)$ rounds.
    Therefore, $T(n) = O(n) + 2 T(n/2) = O(n \log(n))$.
\end{proof}

\subsection{Separating $\AC(G)$ From Other Parameters}\label{sec:separating AC(G) from other params}

In this section we provide several additional examples.
These examples separate $\AC(G)$ from other parameters of graphs.
Our first example shows a graph with low diameter, low clique cover number
(that is, $\overline{G}$ has low chromatic number), such that $\AC(G)$ is large.
\begin{proposition}{\em($\AC$ of the barbell graph):}\label{prop:two cliques}
    Let $G$ be the barbell graph. That is, $G$ consists of
    two cliques of size $n$ connected by a single edge, called bridge.
    Then $\AC(G) = \Theta(n)$.
\end{proposition}

\begin{proof}
    The upper bound follows from Hamiltonicity of $G$ (see Corollary~\ref{cor:hamiltonian}).
    For the lower bound, denote the vertices of the two cliques by $A$ and $B$,
    and denote the bridge by $(a_0,b_0)$, where $a_0 \in A$ and $b_0 \in B$.
    Then, in any strategy for acquaintance either all agents from $A$
    visited $a_0$, or all agents from $B$ visited $b_0$,
    and the proposition follows.
\end{proof}

A more interesting example shows the existence of a Ramsey graph $G$ with $\AC(G) = 1$,
where by Ramsey graphs we refer to graphs that contains neither a clique
nor an independent set of logarithmic size.
For more details regarding graphs with $\AC(G) = 1$ see Section~\ref{sec:AC(G)=1}.
\begin{proposition}{\em(Ramsey graph with $\AC(G)=1$):}\label{prop:ramsey AC=1}
    There is a graph $G$ on $n$ vertices that contains neither a clique
    nor an independent set of size $\Omega(\log(n))$ such that $\AC(G) = 1$.
\end{proposition}

\begin{proof}
    Let $H=(U = \{u_1,\dots,u_{n/2}\},F)$ be a Ramsey graph on $n/2$ vertices
    that contains neither a clique nor an independent set of size $O(\log(n))$.
    We construct $G=(V,E)$ as follows.
    The vertices of $G$ are two copies of $U$, i.e., $V = \{u_1,\dots,u_{n/2}\} \cup \{v_1,\dots,v_{n/2}\}$.
    The edges of $G$ are the following.
    \begin{enumerate}
    \item
        The vertices $\{u_1,\dots,u_{n/2}\}$ induce a copy of $H$.
        That is, $(u_i,u_j) \in E$ if and only if $(u_i,u_j) \in F$.
    \item
        The vertices $\{v_1,\dots,v_{n/2}\}$ induce the complement of $H$.
        That is, we set $(v_i,v_j) \in E$ if and only if $(u_i,u_j) \notin F$.
    \item
        For each $i,j \in [n/2]$ we have $(u_i,v_j) \in E$.
    \end{enumerate}
    By the properties of $H$ it follows that $G$ is also a Ramsey graph.
    Now, it is straightforward to check that the matching
    $M = \{(u_i,v_i) : i \in [n/2]\}$ is a 1-round strategy for acquaintance.
\end{proof}

The proof of Proposition~\ref{prop:expander} may suggest that a small
routing number (as defined by Alon et~al.~\cite{ACG}) implies fast acquaintance time.
The following example shows separation between the two parameters for the
complete bipartite graph $K_{n,n}$. It was shown in~\cite{ACG} that in $K_{n,n}$
for any permutation of the vertices $\sigma :V \to V$ the agents can be
routed from $v \in V$ to the destination $\sigma(v)$ in 4 rounds. We prove
next that $\AC(K_{n,n}) = \Theta(\log(n))$.
\begin{proposition}{\em($\AC$ of $K_{n,n}$):}\label{prop:K_{n,n}}
    Let $n = 2^r$ for some $r \in \N$.
    Let $K_{n,n} = (A,B,E)$ be complete bipartite graph
    with $|A| = |B| = n$.
    Then $\AC(K_{n,n}) = \log_2(n)$.
\end{proposition}

\begin{proof}
    Assign each agent a string $x = (x_0,x_1,\dots,x_r) \in \{0,1\}^{r+1}$ such that all agents
    who started on the same side have the same first bit $x_0$.
    We now describe an $r$-rounds strategy for acquaintance.
    In the $i$'th round move all agents with $x_i=0$ to $A$
    and all agents with $x_i=1$ to $B$.
    Now if two agents are assigned strings $x$ and $x'$ such that
    $x_i \neq x'_i$, then in the $i$'th round they will be on different sides
    of the graph, and hence will be acquainted.

    We now claim that $r$ rounds are also necessary. Indeed,
    suppose we have a $t$-rounds strategy for acquaintance.
    Assign each agent a string $x = (x_0,x_1,\dots,x_t) \in \{0,1\}^{t+1}$,
    where $x_i = 0$ for $i \leq t$ if and only if in the $i$'th round
    the agent was in $A$. Note that two agents met during the $t$ rounds
    if and only if their strings are different. This implies $2^{t+1} \geq 2n$,
    and thus $t \geq r$, as required.
\end{proof}

\section{The Range of $\AC(G)$}\label{sec:range of AC(G)}

In this section we provide examples of families of graphs on $n$ vertices
whose acquaintance time ranges from constant to $n^{1.5}$.

\begin{theorem}\label{thm:range of AC}
    For all $n \in \N$ and for all positive integers $k \leq n^{1.5}$
    there exists an $n$-vertex graph $G$ such that
    $k/c \leq \AC(G) \leq c \cdot k$ for some universal constant $c \geq 1$.
\end{theorem}

The proof of the theorem is divided into two parts.
In Proposition~\ref{prop:clique+matching} we take care of $k \leq n$,
and Proposition~\ref{prop:octopus} takes care of $n \leq k \leq n^{1.5}$.

\begin{proposition}\label{prop:clique+matching}
    For all $n \in \N$ and for all positive integers $k \leq n$
    there exists an $n$-vertex graph $G$ such that
    $k/c \leq \AC(G) \leq c \cdot k$ for some universal constant $c \geq 1$.
\end{proposition}
\begin{proof}
    In order to prove the proposition for $n \in \N$ such that $n = 0 \pmod k$
    consider the graph $G_{k,\ell} = (V,E)$ with vertices
    $V = \{v_{i,j} : i \in [k], j \in [\ell] \}$,
    where the vertices $\{v_{i,j} : j \in [\ell]\}$ form a
    clique for all $i \in [k]$, and, in addition, for every $i,i' \in [k]$
    such that $|i-i'| = 1$ we have $(v_{i,j},v_{i',j}) \in E$ for all $j \in [\ell]$.
    That is, the vertices are divided into $k$ cliques each of size $\ell$,
    and the edges between adjacent cliques form a perfect matching.

    We claim that $\AC(G_{k,\ell}) = \Theta(k)$.
    For a lower bound $\diam(G_{k,\ell}) = k$
    implies that $\AC(G_{k,\ell}) = \Omega(k)$.
    For an upper bound consider first the case of $k=2$;
    that is, the graph consisting of two disjoint cliques each of size $\ell$,
    with $\ell$ edges between them forming a perfect matching.
    Then $\AC(G_{2,\ell}) = O(1)$, which can be witnessed
    by swapping $\ell/2$ vertices in one clique with
    $\ell/2$ vertices in the other clique a constant number of times.

    The bound $\AC(G_{k,\ell}) = O(k)$ is obtained by using a strategy
    similar to the one for $P_k$ described in Proposition~\ref{prop:path},
    where we consider each clique as a single block, and each swap in $P_k$
    corresponds to a swap of the blocks, rather than single vertices.
    The only difference is the fact that even if two blocks of size
    $\ell$ are adjacent, it does not imply that all the $2\ell$ agents
    in the two blocks have met. In order to make them meet we apply the
    $O(1)$-rounds strategy above for the $G_{2,\ell}$ graph. This completes
    the proof for $n = 0 \pmod k$.

    In order to generalize the example above to general $n \in \N$
    let $\ell = \floor{n/k}$, and let $r = n \pmod k$.
    Consider the graph $G$ obtained from $G_{k,\ell}$ by adding
    to it an $r$-clique and connecting it to $r$ vertices in the last
    clique by a matching of size $r$. That is, the graph consists of
    $k$ cliques of size $\ell$ and another  clique of size $r$ with
    matchings of maximal size between consecutive cliques.
    The lower bound of $\AC(G) \geq \floor{\diam(G)/2} = \Omega(k)$ still holds, whereas for
    the upper bound we can first move the agents from the $r$-clique to meet all
    other vertices by visiting all cliques and return them back in $2k$ rounds,
    and then apply the strategy for $G_{k,\ell}$. This gives us
    $\AC(G) = \Theta(k)$.
\end{proof}

\begin{proposition}\label{prop:octopus}
    For all $n \in \N$ and for all positive integers $k \in [n,n^{1.5}]$
    there exists an $n$-vertex graph $G$ such that
    $k/c \leq \AC(G) \leq c \cdot k$ for some universal constant $c \geq 1$.
\end{proposition}
\begin{proof}
    Consider the graph $O_{r,\ell}$ that consists of $r$ cliques of size $\ell$,
    another vertex $z$ called the center, and in each clique
    one of the vertices is connected to the center.
    We claim $\AC(O_{r,\ell}) = \Theta(\min(n \ell,n r))$.
    Since the total number of vertices in the graph is $n = r \ell + 1$,
    by choosing $r \approx k/n$ and $\ell \approx n^2/k$ we will get that
    $\AC(O_{r,\ell}) = \Theta(k)$, as required.

    In order to prove an upper bound of $O(n r)$ note that solving
    the acquaintance problem on $O_{r,\ell}$ can be reduced to solving ${r \choose 2}$
    problems of Hamiltonian graphs of size $2\ell + 1$, where each problem corresponds
    to a pair of cliques with the center $z$. By Hamiltonicity each such problem
    is solved in $O(\ell)$ rounds.

    In order to prove an upper bound of $O(n \ell)$ we can bring every agent
    to the center, and all other agents will meet him in $O(\ell)$ rounds, using the
    vertices connected to $z$.

    For the lower bound define for every agent $p_i$ and every time $t \in \N$
    the variable $\varphi_t(p_i)$ to be the number of agents that $p_i$ has
    met up to time $t$. Note that for $h = \AC(O_{r,\ell})$ there is a strategy
    such that $\sum_{i \in [n]}\varphi_{h}(p_i) = n \cdot (n-1)$,
    since every agent met every other agent up to time $h$.
    On the other hand, in each time $t$ the sum $\sum_{i \in [n]} \varphi_t(p_i)$ increases
    by at most $2r+\ell$, as the only agents who could potentially affect the sum
    are those who moved to the center (contributing at most $r$ to the sum),
    an agent who moved from the center to one of the cliques (contributing at most $\ell$ to the sum),
    and the $r$ neighbors of the center (each contributing at most 1 to the sum).
    This implies a lower bound of $\AC(O_{r,\ell}) \cdot (2r + \ell) = \Omega(n^2)$,
    as required.

    In order to construct an $n$-vertex graph for general $n$ and $k$
    as in the assumption, take $r = \floor{\frac{k}{n}}$,
    $\ell = \floor{\frac{n-1}{r}}$, and $t = n-1 \pmod{\ell}$.
    Consider the $n$-vertex graph analogous to the construction above,
    that consists of $r-t$ cliques of size $\ell$, $t$ cliques of size $\ell+1$,
    and a center connected to one of the vertices is each clique.
    The argument above proves Proposition~\ref{prop:octopus}.
\end{proof}

Building on the lower bound in the proof of Proposition~\ref{prop:octopus}
we show that bottlenecks in graphs imply high acquaintance time.

\begin{proposition}\label{prop:bottlenecks}
    Let $G = (V,E)$ be a graph with $n$ vertices.
    Suppose there is a subset of vertices $S \seq V$
    such that when removing $S$ from $G$
    each connected component in the remaining graph is of size at most $\ell$.
    Then $\AC(G) = \Omega(\frac{{n \choose 2} - |E|}{|S|\cdot\ell+\sum_{s \in S}\deg(s)})$.
\end{proposition}

\begin{proof}
    Denote by $G[V \setminus S]$ the graph obtained by removing from $G$
    the vertices in $S$ and the edges touching them.
    Define for every agent $p_i$ and every time $t \in \N$
    the set $\varphi_t(p_i) \seq \{p_j : j \in [n]\}$ to contain all agents that $p_i$
    has met up to time $t$, as well as all agents who shared a connected component in $G[V \setminus S]$ with $p_i$ up to time $t$.
    By definition of $\AC$ for $h = \AC(G)$ there is a strategy such that
    $\sum_{i \in [n]}|\varphi_{h}(p_i)| = n \cdot (n-1)$,
    since every agent met every other agent up to time $h$.
	Note that in the $t$'th round the increment to $\varphi_t(p_i)$
    compared to $\varphi_{t-1}(p_i)$ is either because $p_i$ entered $S$
    and met new agents in $S$ and in its connected components of $G[V \setminus S]$,
    or because an agent left $S$ and entered one of the connected components.
    Thus, in each time $t$ the sum $\sum_{i \in [n]} |\varphi_t(p_i)|$
    increases by at most $|S|\cdot \ell+\sum_{s \in S}\deg(s)$, where $|S|\cdot \ell$
    upper bounds the number of meetings that were added because of agents moving out of $S$,
    while the value $\sum_{s \in S}\deg(s)$ bounds the number of meetings that are accounted
    for by agents that entered $S$ in round $t$. This implies a lower bound of
    $\Omega(\frac{{n \choose 2} - |E|}{|S|\cdot \ell+\sum_{s \in S}\deg(s)})$,
    which completes the proof of Proposition~\ref{prop:bottlenecks}.
\end{proof}

Next, we show that for every graph $G$ with $n$ vertices
the acquaintance time is in fact asymptotically smaller than the trivial bound of $2n^2$.
Specifically, we prove the following theorem.%
\footnote{As mentioned in the introduction, this result has been superseded
by Angel and Shinkar in~\cite{AS}.}
\begin{theorem}\label{thm:AC-bound}
    For every graph $G$ with $n$ vertices it holds that $\AC(G) = O\left(\frac{n^2}{\log(n)/\log\log(n)}\right)$.
\end{theorem}

The proof of Theorem~\ref{thm:AC-bound} relies on the following two claims.

\begin{claim}\label{claim:AC-bound long path}
    Let $G$ be a graph with $n$ vertices. If $G$ contains a simple path of length $\ell$,
    then $\AC(G) = O(n^2/\ell)$.
\end{claim}

\begin{claim}\label{claim:AC-bound maxdeg}
    Let $G$ be a graph with $n$ vertices. If $G$ has a vertex of degree $\Delta$,
    then $\AC(G) = O(n^2/\Delta)$.
\end{claim}

We postpone the proofs of the claims until later,
and show how to deduce Theorem~\ref{thm:AC-bound} from them.

\begin{proof-of-thm}{\ref{thm:AC-bound}}
    Let $k = \Theta\left(\frac{\log(n)}{\log\log(n)}\right)$
    be the largest integer such that $k^k \leq n$. For such a choice of $k$ the graph $G$
    either contains a simple path of length $k$, or it contains a vertex of degree at
    least $k$. In the former case by Claim~\ref{claim:AC-bound long path} we have
    $\AC(G) = O(n^2/k)$. In the latter case we use Claim~\ref{claim:AC-bound maxdeg}
    to conclude that $\AC(G) = O(n^2/k)$. The theorem follows.
\end{proof-of-thm}

\noindent
We now prove Claims~\ref{claim:AC-bound long path} and~\ref{claim:AC-bound maxdeg}.

\begin{proofof}{Claim~\ref{claim:AC-bound long path}}
    Assume without loss of generality that $\ell$ is the length
    of the longest simple path in $G$. Then, in particular, we
    have $\dist(u,v) \leq \ell$ for every two vertices $u,v \in V$.
    We shall also assume that $G$
    is a tree that contains a path of length $\ell$ (if not, apply the
    argument below on a spanning tree of $G$, which is enough as $\AC(G)$
    is upper bounded by $\AC$ of its spanning tree).

    In order to prove the claim we apply Claim~\ref{claim:routing subgraphs}
    together with Proposition~\ref{prop:path}
    similarly to the proof of Proposition~\ref{prop:expander}.
    Divide the agents into $c = O(n/\ell)$ subsets $C_1,\dots,C_c$ of
    size at most $\floor{\ell/2}$ each.
    For every pair $i,j \in [c]$ we use Claim~\ref{claim:routing subgraphs}
    to route the agents from the two subsets $C_i \cup C_j$ to a path of length $\ell$,
    and then apply the strategy from Proposition~\ref{prop:path}
    so that every pair of agents from $C_i \cup C_j$ meet.
    By repeating this strategy for every $i,j \in [c]$, we make sure
    that every pair of agents on $G$ meet each other.

    Since $\dist(u,v) \leq \ell$ for every two vertices $u,v \in V$,
    by Claim~\ref{claim:routing subgraphs} the agents $C_i \cup C_j$
    can be routed to a path of length $\ell$ in $O(\ell)$ rounds.
    Then, using the strategy from Proposition~\ref{prop:path}
    every pair of agents from $C_i \cup C_j$ meet in at most $O(\ell)$.
    Therefore, the acquaintance time of $G$ can be upper bounded by
    \[
        \AC(G) \leq {c \choose 2} \cdot O(\ell) = O(n^2/\ell),
    \]
    as required.
\end{proofof}

\begin{proofof}{Claim~\ref{claim:AC-bound maxdeg}}
    Assume without loss of generality that $G$ is a tree rooted
    at a vertex $r$ of degree $\Delta$. (This can be done by considering
    a spanning tree of $G$.) Denote the children of $r$ by $v_1,\dots,v_\Delta$,
    and let $p_1,\dots,p_\Delta$ be the agents originally located at these vertices.
    We claim that there is an $O(n)$-rounds strategy that allows $p_1,\dots,p_\Delta$
    to meet all agents.

    Given such a strategy, we apply it on $G$ repeatedly,
    with $\Delta$ new agents in $v_1,\dots,v_\Delta$ in each iteration.
    The agents can be placed there in $n + 2\Delta$ rounds using
    Claim~\ref{claim:routing subgraphs}.
    Repeating the process we get that
    $AG(G) \leq \sum_{i=1}^{\ceil{n/\Delta}}O(n+\Delta) = O(n^2/\Delta)$, as required.
%

    Next, we describe an $O(n)$-rounds strategy that allows $p_1,\dots,p_\Delta$ to meet
    all agents. For any $1 \leq i \leq \Delta$ consider a subtree $T_i=(V_i,E_i)$
    of $G$ rooted at $v_i$. It is enough to show how the agents $p_1,\dots,p_\Delta$
    can meet all agents from $T_i$ in $O(|T_i|)$ rounds.
    First, let $p_i$ meet all agents in $T_i$ in $O(|T_i|)$ steps and return back to $v_i$.
    This can be done by running $p_i$ along a DFS of $T_i$. It is enough now to find a
    $O(|T_i|)$- rounds strategy that allows all agents of $T_i$ to visit the root $r$.
    This task can be reduced to the routing problem considered in Claim~\ref{claim:routing subgraphs}.
    Specifically, define a tree $T'$ on $2|V_i|+1$ vertices rooted at $r$ that contains the tree $T_i$ with
    additional $|V_i|$ vertices each connected only to the root $r$.

    By Claim~\ref{claim:routing subgraphs} there is a $O(|T_i|)$-rounds strategy
    in $T'$ that routes all agents from the copy of $T_i$ to the additional $|V_i|$ vertices.
    It is easy to see that this strategy
    can be turned into a strategy that allows all agents of $T_i$ to visit in the root $r$
    by disregarding the edges between $r$ and the additional (imaginary) vertices.
    This completes the proof of Claim~\ref{claim:AC-bound maxdeg}.
\end{proofof}

\section{$\NP$-Hardness Results}\label{sec:NP-hardness}

In this section we show that the acquaintance time problem is $\NP$-hard.
Specifically, we prove the following theorem.
\begin{theorem}\label{thm:NP-hardness}
    For every $t \geq 1$ it is $\NP$-hard to distinguish whether
    a given graph $G$ has $\AC(G) \leq t$ or $\AC(G) \geq 2t$.
\end{theorem}

Before actually proving the theorem, let us first see the proof for the special case of $t=1$.

\paragraph{Special case of $t=1$:}
    We start with the following $\NP$-hardness result,
    saying that for a given graph $G$ it is hard to distinguish between graphs
    with small chromatic number and graphs with somewhat large independent set.
    Specifically, Lund and Yanakakis~\cite{LY94} prove the following result.
    \begin{theorem}[{\cite[Theorem 2.8]{LY94}}]
        For every $K \in \N$ sufficiently large
        the following gap problem is $\NP$-hard.
        Given a graph $G = (V,E)$ distinguish between the following two cases:
        \begin{itemize}
        \item
            $\chi_\eq(G) \leq K$; i.e., there exists a $K$-coloring of the vertices of
            $G$ with color classes of size $\frac{|V|}{K}$ each.%
            \footnote{The statement of Theorem 2.8 in~\cite{LY94} says that
            $G$ is $K$-colorable. However, it follows from the proof
            that in fact $G$ is equi-$K$-colorable.}
        \item
            $\alpha(G) \leq \frac{n}{2K}$.%
            \footnote{The statement of Theorem 2.8 in~\cite{LY94} says that
            $\chi(G) \geq 2K$. However, the proof implies that in fact
            $\alpha(G) \leq \frac{n}{2K}$.}
        \end{itemize}
    \end{theorem}
    We construct a reduction from the problem above to the acquaintance time problem,
    that given a graph $G$ outputs a graph $H$ so that (1) if $\chi_\eq(G) \leq K$, then $\AC(H) = 1$,
    and (2) if $\alpha(G) \leq \frac{n}{2K}$, then $\AC(H) \geq 2$.

    Given a graph $G = (V,E)$ with $n$ vertices $V = \{ v_i: i \in [n]\}$,
    the reduction outputs a graph $H=(V',E')$ as follows.
    The graph $H$ contains $|V'| = 2n$ vertices,
    partitioned into two parts $V' = V \cup U$, where
    $|V| = n$ and $U = U_1 \cup \dots \cup U_K$
    with $|U_j| = n/K$ for all $j \in [K]$.
    The vertices $V$ induce the complement graph of $G$.
    For each $j \in [K]$ the vertices of $U_j$ form an independent set.
    In addition, we set edges between every pair of vertices $(v,u) \in V \times U$
    as well as between every pair of vertices in $(u,u') \in U_j \times U_{j'}$ for all $j \neq j'$.
    This completes the description of the reduction.

    \paragraph{Completeness:}
    We first prove the completeness part, namely, if $\chi_\eq(G) = K$, then $\AC(H) = 1$.
    Suppose that the color classes of $G$ are $V = V_1 \cup \dots \cup V_K$
    with $|V_j| = n/K$ for all $j \in [K]$. Note that each color class $V_j$ induces a clique in $H$.
    Consider the matching that
    for each $j \in [K]$ swaps the agents in $U_j$ with the agents in $V_j$.
    (This is possible since by the assumption $|V_j| = \frac{n}{K}$,
    and all vertices of $U_j$ are connected to all vertices of $V_j$.)
    In order to verify that such matching allows every pair of agents to meet each other,
    let us denote by $p_v$ the agent sitting originally in vertex $v$.
    Note that before the swap all pairs listed below have already met.
    \begin{enumerate}
        \item
            For all $j \in [K]$ and every $v,v' \in V_j$ the pair of agents $(p_{v},p_{v'})$ have met.
        \item
            For all $j \neq j'$ and for every $u \in U_j$, $u' \in U_{j'}$
            the pair of agents $(p_{u},p_{u'})$ have met.
        \item
            For all $v \in V$ and $u \in U$
            the pair of agents $(p_{v},p_{u})$ have met.
    \end{enumerate}
    After the swap the following pairs meet.
    \begin{enumerate}
        \item
            For all $j \neq j'$ and for every $v \in V_j$, $v' \in V_{j'}$
            the agents $p_{v}$ and $p_{v'}$
            meet using an edge between $U_j$ and $U_j'$.
        \item
            For all $j \in [K]$ and every $u,u' \in V_j$
            the agents $p_{u}$ and $p_{u'}$ meet using an edge in $V_j$.
    \end{enumerate}
    This completes the completeness part.

    \paragraph{Soundness:}
    For the soundness part assume that $\AC(H) = 1$. We claim that
    $\alpha(G) > 1/2K$. Note first that if there is a single matching
    that allows all agents to meet, then for every $j \in [K]$
    all but at most $K$ agents from $U_j$ must have been moved by the matching to $V$.
    (This holds since $U$ does not contain $K+1$ clique.)
    Moreover, all the agents from $U_j$ who moved to $V$ must have moved
    to a clique induced by $V$.
    This implies that $V$ contains a clique of size at least $n/K - K > n/2K$,
    which implies that $\alpha(G) > n/2K$.

    \medskip
    This completes the proof of Theorem~\ref{thm:NP-hardness} for the special case of $t=1$.
    The proof of Theorem~\ref{thm:NP-hardness} for general $t \geq 2$ is quite similar,
    although it requires some additional technical details.

\begin{proofof}{Theorem~\ref{thm:NP-hardness}}
    We start with the following $\NP$-hardness result due to Khot~\cite{Khot01}, saying
    that for a given graph $G$ it is hard to distinguish between graphs with small chromatic
    number and graphs with small independent set. Specifically, Khot proves the following result.
    \begin{theorem}[{\cite[Theorem 1.6]{Khot01}}]\label{thm:Khot's hardness}
        For every $t \in \N$, and for every $K \in \N$ sufficiently large
        (it is enough to take $K = 2^{O(t)}$)
        the following gap problem is $\NP$-hard.
        Given a graph $G = (V,E)$ distinguish between the following two cases:
        \begin{itemize}
        \item
            $\chi_\eq(G) \leq K$.
        \item
            $\alpha(G) \leq \frac{n}{4t^{2t+1}K^{2t}}$.
        \end{itemize}
    \end{theorem}
    We construct a reduction from the problem above to the acquaintance time problem,
    that given a graph $G$ outputs a graph $H$ so that (1) if $\chi_\eq(G) \leq K$, then $\AC(H) \leq t$,
    and (2) if $\alpha(G) \leq \frac{n}{4t^{2t+1}K^{2t}}$, then $\AC(H) \geq 2t$.

    Given a graph $G = (V,E)$ with $n$ vertices, the reduction $r(G)$ outputs a graph $H=(V',E')$ as follows.
    The graph $H$ contains $|V'| = (t+1)n$ vertices,
    partitioned into two parts $V' = V \cup U$, where
    $|V| = n$ and $U = \cup_{i\in [t],j \in [K]} U_{i,j}$
    with $|U_{i,j}| = n/K$ for all $i \in [t],j \in [K]$.
    The vertices $V$ induce the complement graph of $G$.
    For each $i \in [t],j \in [K]$ the vertices of $U_{i,j}$ form an independent set.
    In addition, we set edges between every pair of vertices $(v,u) \in V \times U$
    as well as between every pair of vertices in $(u,u') \in U_{i,j} \times U_{i',j'}$ for all $(i,j) \neq (i',j')$.
    This completes the description of the reduction.

    \paragraph{Completeness:}
    We first prove the completeness part, namely, if $\chi_\eq(G) = K$, then $\AC(H) \leq t$.
    Suppose that the color classes of $G$ are $V = V_1 \cup \dots \cup V_K$
    with $|V_j| = n/K$ for all $j \in [K]$.
    Note that each color class $V_j$ induces a clique in $H$.
    We show that $\AC(H) \leq t$, which can be achieved as follows:
    For all $i \in [t]$ in the $i$'th round the agents located in vertices $U_{i,j}$
    swap places with the agents in $V_j$ for all $j \in [K]$.
    (This is possible since by the assumption $|U_{i,j}| = |V_j| = \frac{n}{K}$,
    and all vertices of $U_{i,j}$ are connected to all vertices of $V_j$.)

    We next verify that this strategy allows every pair of agents to meet each other.
    Indeed, denoting by $p_v$ the agent sitting originally in vertex $v$
    the only pairs who did not met each other before the first round are
    contained in the following two classes:
    \begin{enumerate}
        \item
            For all $j \neq j' \in [K]$ and for every $v \in V_j$, $v' \in V_{j'}$
            the pair of agents $(p_{v},p_{v'})$.
        \item
            For each $i \in [t]$ and $j \in [K]$ and every $u,u' \in U_{i,j}$
            the pair of agents $(p_{u},p_{u'})$.
    \end{enumerate}
    Then, when the agents move along the prescribed matchings, the pairs
    from the first class meet after the first round.
    And for each round $i \in [t]$ the pairs from the second class
    that correspond to $u,u' \in U_{i,j}$ for some $j \in [K]$
    meet after the $i$'th round.
    This proves the completeness part of the reduction.

    \medskip
    \paragraph{Soundness:}
    For the soundness part assume that $\AC(H) \leq 2t-1$, and consider the
    corresponding $(2t-1)$-rounds strategy for acquaintance in $G$.
    By a counting argument there are $\frac{n}{2}$ agents who
    originally were located in $U$ and visited $V$ at most once.
    By averaging, there are $\frac{n}{2(2t-1)}$ agents who either never visited $V$
    or visited $V$ simultaneously, and this was their only visit to $V$.
    Let us denote this set of agents by $P_0$.
    The following claim completes the proof of the soundness part.
    \begin{claim}
        Let $P_0$ be a set of agents of size $\frac{n}{2(2t-1)}$.
        Suppose they visited $V$ at most once simultaneously,
        and visited $U$ at most $2t$ times.
        If every pair of agents from $P_0$ met each other
        during these rounds, then $\alpha(G) \geq \frac{n}{2(2t-1)(tK)^{2t}} > \frac{n}{4t^{2t+1}K^{2t}}$.
    \end{claim}
    \begin{proof}
        Let us assume for concreteness that $P_0$ stayed in $U$ until the last round,
        and then moved to $V$.
        Associate with each agent a sequence of sets $U_{i,j}$ of length $2t$
        which he visited during the first $2t-1$ rounds.
        This defines a natural partition of $P_0$ into $(tK)^{2t}$
        clusters, where the agents are in the same cluster if and only if
        they have the same sequence.
        That is, the agents meet each other if and only if they are in different clusters.
        Thus, at least one of the clusters is of size at least $\frac{|P_0|}{(tK)^{2t}} = \frac{n}{2(2t-1)(tK)^{2t}}$.
        If we assume that every pair of agents from $P_0$ met each other eventually,
        then it must be the case that in the last round each cluster moved to some
        clique in $V$, and in particular $\overline{G}$ contains a clique of size $\frac{n}{2(2t-1)(tK)^{2t}}$.
        The claim follows.
    \end{proof}

    We have shown a reduction from the coloring problem to the acquaintance time problem,
    that given a graph $G$ outputs a graph $H$ so that (1) if $\chi_\eq(G) \leq K$, then $\AC(H) \leq t$,
    and (2) if $\alpha(G) \leq \frac{n}{4t^{2t+1}K^{2t}}$, then $\AC(H) \geq 2t$.
    This completes the proof of Theorem~\ref{thm:NP-hardness}.
\end{proofof}

\subsection{Towards stronger hardness results}

We conjecture that, in fact, a stronger hardness result holds,
compared to the one stated in Theorem~\ref{thm:NP-hardness}.
\begin{conjecture}\label{conj:NP-hardness 1 vs t}
    For every constant $t \in \N$ it is $\NP$-hard to decide
    whether a given graph $G$ has $\AC(G) = 1$ or $\AC(G) \geq t$.
\end{conjecture}

Below we describe a gap problem similar in spirit to the hardness
results of Lund and Yanakakis and that of Khot whose $\NP$-hardness
implies Conjecture~\ref{conj:NP-hardness 1 vs t}.
In order to describe the gap problem we need the following definition.
\begin{definition}
    Let $t \in \N$ and $\beta > 0$.
    A graph $G = (V,E)$ is said to be {\em $(\beta,t)$-intersecting}
    if for every $t$ subsets (not necessarily disjoint) of the vertices $S_1,\dots,S_t \seq V$
    of size $\beta n$ and for every $t$ bijections $\pi_i : S_i \to [\beta n]$
    there exist $j,k \in [\beta n]$ such that
    all pre-images of the pair $(j,k)$ are edges in $E$,
    i.e., for all $i \in [t]$ it holds that $(\pi_i^{-1}(j),\pi_i^{-1}(k)) \in E$.
\end{definition}

    Note that a graph $G$ is $(\beta,1)$-intersecting if and only if
    $G$ does not contain an independent set of size $\beta n$.
    In addition, note that if $G$ is $(\beta,t)$-intersecting
    then it is also $(\beta',t')$-intersecting for $\beta' \geq \beta$ and $t' \leq t$,
    and in particular $\alpha(G) < \beta$.

    We remark without proof that the problem of deciding whether a given graph $G$
    is $(\beta,t)$-intersecting is $\coNP$-complete.
    We make the following conjecture regarding $\NP$-hardness of distinguishing between
    graphs with small chromatic number and $(\beta,t)$-intersecting graphs.

\begin{conjecture}\label{conj:(a,t)-intersecting}
    For every $t \in \N$ and for all $K \in \N$ sufficiently large
    it is $\NP$-hard to distinguish between the following two cases
    for a given graph $G = (V,E)$:
    \begin{itemize}
    \item
        $\chi_\eq(G) \leq K$.
    \item
        The graph $G$ is $(1/K^t,t)$-intersecting.
    \end{itemize}
\end{conjecture}

\begin{remark}
    Conjecture~\ref{conj:(a,t)-intersecting} does not seem to follow immediately
    from the result of Khot stated in Theorem~\ref{thm:Khot's hardness}.
    One reason for that is due to the fact that Khot's hard instances for the problem
    are bounded degree graphs, and we suspect that such graphs cannot be $(\beta,t)$-intersecting
    for arbitrarily small $\beta>0$ even in the case of $t=2$.
\end{remark}

\begin{theorem}
    Conjecture~\ref{conj:(a,t)-intersecting} implies
    Conjecture~\ref{conj:NP-hardness 1 vs t}.
\end{theorem}

The proof of this implication is analogous to
the proof of Theorem~\ref{thm:NP-hardness}, and we omit it.
The reduction is exactly the same as described in the proof
of Theorem~\ref{thm:NP-hardness} for the special case of $t=1$.
The analysis is similar to the proof of Theorem~\ref{thm:NP-hardness} for general $t \geq 1$,
where instead of using the assumption that $\alpha(G)$ is small
we use the stronger assumption in the NO-case
of Conjecture~\ref{conj:(a,t)-intersecting}.

\section{Graphs with $\AC(G)=1$}\label{sec:AC(G)=1}

In this section we study graphs whose acquaintance time equals 1.
We state some structural results for such graphs, and use them to
give efficient approximation algorithms for $\AC$ on them.
Specifically, for a graph $G$ with $\AC(G)=1$ and a constant $c$ we give a deterministic
algorithm that returns an $n/c$ strategy for acquaintance in $G$
whose running time is $n^{c+ O(1)}$. We also give a randomized
polynomial time algorithm that returns an $O(\log(n))$ strategy for such graphs.
These results appear in Section~\ref{sec:AC=1 algorithms}.

\begin{definition}\label{def:ac=1-structure}
    Let $G=(V,E)$ be a graph, and let
    $V = A \cup B \cup C$ be a partition of the vertices
    with $A = \{a_i\}_{i=1}^k$ and $B = \{b_i\}_{i=1}^k$ for some $k \in \N$.
    The tuple $(A,B,C)$ is called a {\em one-matching-witness for $G$}
    if it satisfies the following conditions.
    \begin{enumerate}
    \item\label{item1:AC=1}
        $(a_i,b_i) \in E$ for all $i \in [k]$.
    \item\label{item2:AC=1}
        Either $(a_i,b_j) \in E$ or $(a_j,b_i) \in E$ for all $i \neq j \in [k]$.
    \item\label{item3:AC=1}
        Either $(a_i,a_j) \in E$ or $(b_i,b_j) \in E$ for all $i \neq j \in [k]$.
    \item\label{item4:AC=1}
        The vertices of $C$ induce a clique in $G$.
    \item\label{item5:AC=1}
        For all $c \in C$ and for all $i \in [k]$ we have either $(c,a_i) \in E$ or $(c,b_i) \in E$.
    \end{enumerate}
    (In items \ref{item2:AC=1}, \ref{item3:AC=1}, and \ref{item5:AC=1}
    the either-or condition is not exclusive.)
\end{definition}

\begin{claim}\label{claim:AC=1}
    A graph $G = (V,E)$ satisfies $\AC(G)=1$
    if and only if it has a one-matching-witness.
\end{claim}
\begin{proof}
    Suppose first that $\AC(G)=1$, and let $M = \{(a_1,b_1),\dots,(a_k,b_k)\}$
    be a matching that witnesses the assertion $\AC(G)=1$.
    Let $A = \{a_i\}_{i=1}^k$, $B = \{b_i\}_{i=1}^k$, and $C = V \setminus (A \cup B)$,
    Then $(A,B,C)$ is a one-matching-witness for $G$.

    For the other direction, if $(A,B,C)$ is a one-matching-witness for $G$, then the matching
    $M = \{(a_1,b_1),\dots,(a_k,b_k)\}$ is 1-round strategy for acquaintance in $G$.
\end{proof}

The following corollary is immediate from Claim~\ref{claim:AC=1}.
\begin{corollary}\label{cor:AC=1}
    Let $G = (V,E)$ be an $n$-vertex graph that satisfies $\AC(G)=1$,
    and suppose that $(A = \{a_i\}_{i=1}^k,B = \{b_i\}_{i=1}^k,C)$
    is a one-matching-witness for $G$. Then,
    \begin{enumerate}
    \item\label{item1:corAC=1}
        For all $i \in [k]$ it holds that $\deg(a_i) + \deg(b_i) \geq 2k + |C| = n$.
    \item\label{item2:corAC=1}
        For all $c \in C$ we have $\deg(c) \geq k+|C|-1 \geq \floor{n/2}$.
    \item\label{item3:corAC=1}
        There are at least $\floor{n/2}$ vertices $v \in V$
        with $\deg(v) \geq \floor{n/2}$.
    \end{enumerate}
\end{corollary}

\begin{claim}\label{claim:matching small to large}
Let $G = (V,E)$ be a graph with $n$ vertices that satisfies $\AC(G)=1$,
and let $U \seq V$ be the set of vertices $v \in V$ such that $\deg(v) \geq \floor{n/2}$.
Then, for every $W \seq V \setminus U$
there exists a matching of size $|W|$ between $U$ and $W$.
\end{claim}

\begin{proof}
    Let $(A,B,C)$ be a one-matching-witness for $G$.
    Note that by Corollary~\ref{cor:AC=1} Item~\ref{item2:corAC=1}
    we have $C \seq U$, and thus $W \seq A \cup B$.
    By Item~\ref{item1:corAC=1} of Corollary~\ref{cor:AC=1} for every $i \in [k]$
    it holds that $\deg(a_i) + \deg(b_i) \geq n$, and thus either $a_i$ or $b_i$
    belongs to $U$. Therefore, the required matching is given by
    $M = \{(a_i, b_i) : \mbox{$i \in [k]$ such that either $a_i \in W$ or $b_i \in W$}\}$.
\end{proof}

The following proposition gives additional details on the structure of graphs with $\AC(G)=1$.
It will be used later for the analysis of a (randomized) approximation algorithm
for acquaintance in such graphs (see Theorem~\ref{thm:1-vs-log(n) in RP}).

\begin{proposition}\label{proposition:E[N(u),N(v)]}
    Let $G = (V,E)$ be a graph with $n$ vertices that satisfies $\AC(G)=1$,
    and let $u,v \in V$ be two vertices of degree at least $n/2$.
    Then, either $|N(u) \cap N(v)| = \Omega(n)$ or $|E[N(u),N(v)]| = \Omega(n^2)$,
    where $E[N(u),N(v)] = \{(a,b)\in E : a \in N(u), b \in N(v)\}$
    denotes the set of edges between $N(u)$ and $N(v)$.
\end{proposition}

\begin{proof}
    If $|N(u) \cap N(v)| \geq 0.1n$, then we are done.
    Assume now that $|N(u) \cap N(v)| < 0.1n$.
    Therefore $|N(u) \cup N(v))| > 0.9n$, as $|N(u)| + |N(v)| \geq n$.
    Define two disjoint sets $N'(u) = N(u) \setminus N(v)$ and $N'(v) = N(v) \setminus N(u)$,
    and note that by disjointness we have $|N'(u)| \geq 0.4n$ and $|N'(v)| \geq 0.4n$.
    It suffices to prove that $|E[N'(u),N'(v)]| = \Omega(n^2)$.

    Suppose that $(A = \{a_i\}_{i=1}^k,B = \{b_i\}_{i=1}^k,C)$
    is a one-matching-witness for $G$.
    Consider the indices $I = \{i \in [k] : a_i,b_i \in N'(u)\cup N'(v)\}$,
    and define a partition $I = I_u \cup I_v \cup I_{u,v}$,
    where $I_u = \{i \in [k] : a_i,b_i \in N'(u)\}$,
    $I_v = \{i \in [k] : a_i,b_i \in N'(v)\}$,
    and $I_{u,v} = I \setminus (I_u \cup I_v)$.
    Also, define $C_u = C \cap N'(u)$, and $C_v = C \cap N'(v)$.
    Note that $|N'(u)| = |C_u| + 2|I_u| + |I_{u,v}|$,
    and analogously $|N'(v)| = |C_v| + 2|I_v| + |I_{u,v}|$.
    Using this partition we have
    \[
        |E[N'(u),N'(v)]| \geq |C_u|\cdot|C_v| + \frac{|I_{u,v}|^2}{2} + |I_u| \cdot |I_v| + |I_u| \cdot |C_v| + |I_v| \cdot |C_u|,
    \]
    where the first term follows from the fact that $C$ induces a clique
    (Definition~\ref{def:ac=1-structure} Item~\ref{item4:AC=1}),
    the second and third terms follow from Item~\ref{item2:AC=1} of Definition~\ref{def:ac=1-structure},
    and the last two terms follow from Item~\ref{item5:AC=1} of Definition~\ref{def:ac=1-structure}.

    Now, if $|I_{u,v}| > 0.2n$, then $|E[N'(u),N'(v)]| \geq \frac{|I_{u,v}|^2}{2} \geq 0.02n^2$, as required.
    Otherwise,  we have $|C_u| + |I_u| \geq 0.1n$ and $|C_v| + |I_v| \geq 0.1n$,
    and therefore $|E[N'(u),N'(v)]| \geq (|C_u| + |I_u|) \cdot (|C_v| + |I_v|) \geq 0.01n^2$,
    as required.
\end{proof}
\subsection{Algorithmic results}\label{sec:AC=1 algorithms}

Recall that (unless $\P = \NP$) there is no polynomial time algorithm that,
when given a graph $G$ with $\AC(G)=1$, finds a 1-round strategy for
acquaintance of $G$. In this section we provide two approximation algorithms
regarding graphs whose acquaintance time equals 1.
In Theorem~\ref{thm:1-vs-n/c in P} we give a deterministic algorithm
that finds an $n/c$-rounds strategy for acquaintance in such graphs whose
running time is $n^{c+O(1)}$.
In Theorem~\ref{thm:1-vs-log(n) in RP} we give a randomized algorithm
that finds an $O(\log(n))$-rounds strategy for acquaintance in such graphs.

We start with the following simple deterministic algorithm.
\begin{proposition}\label{prop:1-vs-n in P}
    There is a deterministic polynomial time algorithm that when given as input
    an $n$-vertex graph $G=(V,E)$ such that $\AC(G)=1$ outputs an
    $n$-rounds strategy for acquaintance in $G$.
\end{proposition}

\begin{proof}
    The algorithm works by taking one agent at a time and finding
    a 1-round strategy that allows this agent to meet all others.
    For each agent $p$ the algorithm works as follows.
    If the location of $p$ is the vertex $v \in V$, then for each
    possible destination $u \in N(v) \cup \{v\}$ for $p$
    the algorithm constructs the bipartite graph $H_{v,u} = (A \cup B, F)$
    where $A = V \setminus (N(v) \cup \{v\})$ and $B = N(u) \setminus \{v\}$,
    and there is an edge $(a,b) \in A \times B$ in $F$ if and only if
    it is contained in $E$. The algorithm then
    looks for a matching of size $|A|$ in $H_{v,u}$.
    Such a matching, if it exists, can be found in polynomial time
    (e.g., using an algorithm for maximum flow).
    We claim below that such a matching, augmented with
    the edge $(v,u)$ if needed (that is, if $v \neq u$),
    gives a 1-round strategy that allows $p$ to meet all other agents.
    Repeating this procedure for all agents gives an
    $n$-rounds strategy for acquaintance in $G$.

    In order to prove correctness of the algorithm
    we claim that for all $v \in V$ there is some $u \in N(v) \cup \{v\}$
    such that the graph $H_{v,u}$ contains a matching of size $|A|$.
    Furthermore, any such matching, augmented with the edge $(v,u)$ if needed,
    gives a 1-round strategy that allows $p$ to meet all other agents.
    Indeed, $H_{v,u}$ contains all edges from the agents who didn't meet
    $p$ before the first round, to the neighbours of $u$.
    Consider a 1-round strategy for acquaintance in $G$, and let $u$
    be the vertex in which $p$ is located after this round.
    This strategy (restricted to the edges of $H_{v,u}$) induces
    a matching of size $|A|$, and so there exists some $u \in N(v) \cup \{v\}$
    such that $H_{v,u}$ contains a matching of size $|A|$.
    To finish the proof of correctness note that any matching of size $|A|$ in $H_{v,u}$
    gives a 1-round strategy that allows $p$ to meet all other agents.
\end{proof}

We modify the proof above to get the following
stronger result.

\begin{theorem}\label{thm:1-vs-n/c in P}
    There is an algorithm that when given as input $c \in \N$ and
    an $n$-vertex graph $G=(V,E)$ with $\AC(G)=1$ outputs an
    $\ceil{n/c}$-rounds strategy for acquaintance in $G$ in time $n^{c+O(1)}$.
\end{theorem}

\begin{proof}
    Modifying the algorithm in the proof of Proposition~\ref{prop:1-vs-n in P}
    we may take $c$ agents at a time, and
    find one matching that allows each of them to meet all agents.
    Given a set $P = \{p_1,\dots,p_c\}$ of agents
    located in the vertices $\{v_1,\dots, v_c\}$ respectively
    the matching is found as follows.

    The algorithm goes over all possible destinations
    $\{ u_i \in N(v_i) \cup \{v_i\} : i \in [c]\}$ for all agents $p_i$.
    If $(v_i,v_j) \notin E$ and $(u_i,u_j) \notin E$ for
    some $i \neq j$, then the agents $p_i$ and $p_j$ do not meet, and
    we skip to the next potential destination.
    Otherwise, define a bipartite graph $H = (A \cup B, F)$ where
    $A = V \setminus ( ( \bigcap_{j \in [c]} N(v_i) )\bigcup \{v_i,u_i : i \in [c]\})$
    and $B = \bigcup_{j \in [c]} N(u_i) \setminus \{v_i,u_i : i \in [c]\}$.
    We add the edge $(a,b) \in A \times B$ to $H$ if the agent
    located originally in $a$ can move to $b$ and meet all agents from $P$
    (who moved in the same round from $v_i$ to $u_i$).
    Formally, we add $(a,b) \in A \times B$ to $H$ iff $(a,b) \in E$
    and for all $i \in [c]$ either $a \in N(v_i)$ or $b \in N(u_i)$ (or both).
    The algorithm then looks for a matching of size $|A|$ in $H$, and
    when found, augments it with $\{(v_i,u_i) : i \in [c]\}$ and adds it to the
    strategy for acquaintance in $G$.

    The correctness is very similar to the argument in the proof of
    Proposition~\ref{prop:1-vs-n in P}, showing that any 1-round strategy
    induces the desired matching in $H$ for some choice of $u_1,\dots,u_c$,
    and that any matching in $H$ gives a 1-round strategy that allows
    each of the agents in $P$ to meet every agent of $G$.
\end{proof}

We now turn to a randomized polynomial time algorithm with the following guarantee.

\begin{theorem}\label{thm:1-vs-log(n) in RP}
    There is a randomized polynomial time algorithm
    which when given a graph $G$ with $\AC(G)=1$
    finds an $O(\log(n))$-rounds strategy for acquaintance in $G$
    with high probability.
\end{theorem}

\begin{proof}
    Let $G = (V,E)$ be an $n$-vertex graph, and
    let $U \seq V$ be the set of vertices of degree at least $\floor{n/2}$.
    By Item~\ref{item3:corAC=1} of Corollary~\ref{cor:AC=1} we have $|U| \geq \floor{n/2}$.
    The following lemma describes a key step in the algorithm.
    \begin{lemma}\label{lemma:log(n) AC in U}
        Let $P_U$ be the agents originally located in $U$. Then, there exists
        a polynomial time randomized algorithm that finds an $O(\log(n))$-rounds
        strategy that makes every two agents in $P_U$ meet with high probability.
    \end{lemma}
    Now, consider all the agents $P$ in $G$. For every subset of $P' \seq P$ of size $|P'| \leq |U|$
    we can use the aforementioned procedure to produce an $O(\log(n))$-rounds strategy
    that allows all agents in $P'$ to meet with high probability.
    Let us partition the agents $P$ into a constant number
    $c = \big\lceil\frac{2|V|}{|U|}\big\rceil$ of disjoint subsets
    $P = P_1 \cup \dots \cup P_c$ with at most $\floor{|U|/2}$ agents each in each $P_i$,
    and apply the procedure to each pair $P_i \cup P_j$ separately.
    By Claim~\ref{claim:matching small to large} we can
    transfer any $P_i \cup P_j$ to $U$ in one step.
    When all pairs have been dealt with, all agents have already met each other.
    This gives us an $O(\log(n))$-rounds strategy
    for the acquaintance problem in graphs with $\AC(G)=1$
    that can be found in randomized polynomial time.
\end{proof}

We return to the proof of Lemma~\ref{lemma:log(n) AC in U}.

\begin{proof-of-lemma}{\ref{lemma:log(n) AC in U}}
    We describe a randomized algorithm that finds an $O(\log(n))$-rounds
    strategy that allows every two agents in $P_U$ to meet.
    Consider the following algorithm for constructing a matching $M$.
    \begin{enumerate}
    \item
        Select a random ordering $\sigma : \{1,\dots,|U|\} \to U$ of $U$.
    \item
        Start with the empty matching $M = \emptyset$.
    \item
        Start with an empty set of vertices $S = \emptyset$.
        The set will include the vertices participating in $M$,
        as well as some of the vertices that will not move.
    \item\label{alg:iteration}
        For each $i=1, \dots, |U|$ do
        \begin{enumerate}
        \item
            Set $u_i=\sigma(i)$.
        \item
            Select a vertex $u_i' \in N(u_i) \cup \{u_i\}$ as follows.
            \begin{enumerate}
            \item
                With probability 0.5 let $u_i' =u_i$.
            \item
                With probability 0.5 pick $u_i' \in N(u_i)$ uniformly at random.
            \end{enumerate}
        \item
            If $u_i \notin S$ and $u'_i \notin S$, then \quad // $(u_i,u'_i)$ will be used in the current step
            \begin{enumerate}
            \item
                $S \leftarrow S \cup \{u_i,u'_i\}$. \quad
            \item
                If $u_i \neq u'_i$, then $M \leftarrow M \cup \{(u_i,u'_i)\}$.
            \end{enumerate}
        \end{enumerate}
        \item
            Output $M$.
    \end{enumerate}
    The following claim bounds the probability that a pair of agents
    in $P_U$ meet after a single step of the algorithm.
    \begin{claim}\label{claim:prob meeting of a pair}
        For every $u,v \in U$, let $p_u$ and $p_v$ be the agents located
        in $u$ and $v$ respectively.
        Then, $\Pr[\mbox{The agents $p_u,p_v$ meet after one step}] \geq c$
        for some absolute constant $c>0$ that does not depend on $n$ or $G$ .
    \end{claim}

    In order to achieve an $O(\log(n))$-rounds
    strategy that allows every two agents in $P_U$ to meet
    apply the matching constructed above, and then return the agents
    to their original positions (by applying the same matching again).
    Repeating this random procedure independently $\ceil{\frac{3\log(n)}{c}}$ times
    will allow every pair of agents to meet with probability at least $1/n^3$.
    Therefore, by a union bound all pairs of agents $p_u,p_v \in P_U$
    will meet with probability at least $1/n$.
    This completes the proof of Lemma~\ref{lemma:log(n) AC in U}.
\end{proof-of-lemma}

\begin{proof-of-claim}{\ref{claim:prob meeting of a pair}}
    We claim first that for every $i \leq |U|$ and for every vertex $w \in U \cup N(U)$
    the probability that in step~\ref{alg:iteration} of the algorithm the vertex $w$
    has been added to $S$ before the $i$'th iteration is upper bounded by $3i/n$. Indeed,
    \begin{eqnarray*}
        \Pr[\exists j < i \mbox{ such that } w \in \{u_{j},u'_{j}\}]
            & \leq & \Pr[w \in \{ u_{j} : j < i\}] + \Pr[w \in \{ u'_{j} : j < i\}] \\
            & \leq & \frac{i}{n} + \sum_{j=1}^i \Pr[u'_{j}=w] \\
            & \leq & \frac{i}{n} + i \cdot 1/\floor{n/2} \\
            \mbox{[for $n>1$]} & \leq & \frac{4i}{n}
    \end{eqnarray*}
    where the bound $\Pr[u'_{j}=w] \leq 1/\floor{n/2}$ follows from the
    assumption that $\deg(u_{j}) \geq \floor{n/2}$ for all $u_{j}\in U$, and hence,
    the probability of picking $u'_{j}$ to be $w$ is $1/\deg(u_{j}) \leq 1/\floor{n/2}$.

    Let $T \in \{2,\dots,|U|\}$ be a parameter to be chosen later.
    Now, let $i \leq |U|$ be the (random) index such that $\sigma(i) = u$,
    and let $j \leq |U|$ be the (random) index such that $\sigma(j) = v$.
    Then,
    \[
        \Pr[ i \leq T \mbox{ and } j \leq T ]
        = \frac{{T \choose 2} \cdot (|U|-2)!}{|U|!}
        = \frac{T(T-1)}{2 \cdot |U| \cdot (|U|-1)} \geq \frac{T^2}{4n^2}.
    \]
    Conditioning on this event, the probability that either $u$ or $u'$
    have been added to $S$ before iteration $i$ is upper bounded by
    $\frac{4i}{n} \leq \frac{4T}{n}$,
    and similarly the probability that either $v$ or $v'$
    have been added in $S$ before iteration $j$ is at most $\frac{4T}{n}$.
    Therefore, with probability at least $\frac{T^2}{4n^2} \cdot (1-\frac{8T}{n})$
    both $(u,u')$ and $(v,v')$ will be used in the current step.
    Therefore,
    \[
        \Pr[\mbox{the agents $p_u,p_v$ meet after one step}]
        \geq \frac{T^2}{4n^2} \cdot (1-\frac{8T}{n}) \cdot \Pr[(u',v') \in E].
    \]
    In order to lower bound $\Pr[(u',v') \in E]$ we use Proposition~\ref{proposition:E[N(u),N(v)]},
    saying that for every two vertices $u,v \in U$ it holds that either
    $N(u) \cap N(v) \geq \alpha n$ or $|E[N(u),N(v)]| \geq \alpha \cdot n^2$ for some
    constant $\alpha > 0$ that does not depend on $n$ or $G$.
    Therefore, for every $u,v \in U$ it holds that $\Pr[(u',v') \in E] \geq \alpha/4$.
    Letting $T = \alpha n/12$ we get that
    $\Pr[\mbox{the agents $p_u,p_v$ meet after one step}] = \Omega(\alpha^3)$,
    as required.
\end{proof-of-claim}

\section{Other Variants and Open Problems}\label{sec:other variants}

There are several variants of the problem that one may consider.
\begin{enumerate}
\item
    The problem of maximizing the number of pairs that meet when
    some predetermined number $t \in \N$ of matchings is allowed.
    Clearly, this problem is also $\NP$-complete, even in the case $t=1$.
\item
    There are two graphs $G$ and $H$. The agents move along
    matchings of $H$, but meet if they share an edge in $G$.
    In particular, this looks natural if $H$ is contained in $G$.
\item
    Instead of choosing a matching in each round,
    one may choose a vertex-disjoint collection of cycles,
    and move agents one step along the cycle. This is a generalization
    of the problem discussed in this paper, where we allow only
    collections of 2-cycles.
\end{enumerate}

\medskip
\noindent
One may also consider a more game-theoretic variant of the problem:
Let $G=(V,E)$ be a fixed graph with one agent sitting in each vertex of $G$.
In each round every agent $p_u$ sitting in a vertex $u \in V$ chooses
a neighbor $u' \in N(v)$ according to some strategy.
Then, for every edge $(v,w) \in E$ the agents $p_v$ and $p_w$
swap places if the choice of the agent $p_v$ was $w$ and
the choice of $p_w$ was $v$.
Suppose that the graph is known, but the agents have no information
regarding their location in the graph
(e.g., $G$ is an unlabeled vertex transitive graph).
Find an optimal strategy for the agents so that
everyone will meet everyone else as quickly as possible.
The question also makes sense in the case where the graph
is not known to the agents.

\medskip
We conclude with a list of open problems.

\begin{problem}
    Find $\AC$ of the Hypercube graph.
    Recall that $\AC(\rm{Hypercube})$ is between $\Omega(n/\log(n))$ and $O(n)$,
    where the lower bound is trivial from the number of edges,
    and the upper bound follows from Hamiltonicity of the graph
    (Corollary~\ref{cor:hamiltonian}).
\end{problem}

\begin{problem}
    Prove Conjecture~\ref{conj:NP-hardness 1 vs t},
    namely, that for every constant $t \in \N$ it is $\NP$-hard to decide
    whether a given graph $G$ has $\AC(G) = 1$ or $\AC(G) \geq t$.
    Recall that it follows from Conjecture~\ref{conj:(a,t)-intersecting}.
\end{problem}

\begin{problem}
    Prove stronger inapproximability results. Is it true that $\AC$ is hard
    to approximate within a factor of $\log(n)$? How about $n^{0.01}$? How
    about $n^{0.99}$? Note that the upper bound $\AC(G) \leq  n^2/\Delta$
    from Claim~\ref{claim:AC-bound maxdeg} together the lower bound of
    $\AC(G) \geq  n/\Delta$ gives an $O(n)$-approximation algorithm
    for the problem.
\end{problem}

\begin{problem}
    Give a polynomial time algorithm that given a graph $G$
    outputs a graph $H$ such that $\AC(H) = f(\AC(G))$ for a
    super-linear function $f:\N \to \N$ (e.g., $f(n) = n^2$).
    Such an algorithm can be useful for hardness
    of approximation results for the $\AC$ problem.
\end{problem}

\begin{problem}
    Derandomize the algorithm given in the proof of
    Theorem~\ref{thm:1-vs-log(n) in RP}.
\end{problem}

\begin{problem}
    Give a structural result regarding graphs with small constant
    values of $\AC(G)$ similar to Claim~\ref{claim:AC=1}.
    Also, is there an efficient $O(\log(n))$-approximation algorithm
    for such graphs?
\end{problem}

\section*{Acknowledgments}
We are grateful to the anonymous referees for many useful comments.
We especially wish to thank the referee who suggested modifying Proposition~\ref{prop:1-vs-n in P}
to obtain Theorem~\ref{thm:1-vs-n/c in P}.
\bibliographystyle{alpha}
\bibliography{ac}

\newcommand{\etalchar}[1]{$^{#1}$}
\begin{thebibliography}{CDG{\etalchar{+}}14}

\bibitem[ACG94]{ACG}
N.~Alon, F.~R.~K. Chung, and R.~L. Graham.
\newblock Routing permutations on graphs via matchings.
\newblock {\em SIAM Journal on Discrete Mathematics}, 7(3):513--530, 1994.

\bibitem[AS13]{AS}
O.~Angel and I.~Shinkar.
\newblock A tight upper bound on acquaintance time of graphs.
\newblock {\em Available from {\em http://arxiv.org/abs/1307.6029}}, 2013.

\bibitem[BHK04]{BHK}
A.~Bj\"{o}rklund, T.~Husfeldt, and S.~Khanna.
\newblock Approximating longest directed paths and cycles.
\newblock In {\em Proceedings of the 31st International Colloquium on Automata,
  Languages and Programming}, pages 222--233, 2004.

\bibitem[BST13]{BST}
I.~Benjamini, I.~Shinkar, and G.~Tsur.
\newblock Acquaintance time of a graph.
\newblock {\em Available from {\em http://arxiv.org/abs/1302.2787}}, 2013.

\bibitem[CDG{\etalchar{+}}14]{CDGKKP}
J.~Czyzowicz, D.~Dereniowski, L.~Gasieniec, R.~Klasing, A.~Kosowski, and
  D.~Pajak.
\newblock Collision-free network exploration.
\newblock In {\em Proceedings of the 11th Latin American Theoretical
  INformatics Symposium}, 2014.

\bibitem[Che09]{Ch}
N.~Chen.
\newblock On the approximability of influence in social networks.
\newblock {\em SIAM Journal on Discrete Mathematics}, 23(5):1400--1415, 2009.

\bibitem[HHL88]{HHL}
S.~T. Hedetniemi, S.~M. Hedetniemi, and A.~Liestman.
\newblock A survey of gossiping and broadcasting in communication networks.
\newblock {\em Networks}, 18(4):319--349, 1988.

\bibitem[Kho01]{Khot01}
S.~Khot.
\newblock Improved inaproximability results for maxclique, chromatic number and
  approximate graph coloring.
\newblock In {\em Proceedings of the 42nd IEEE Symposium on Foundations of
  Computer Science}, pages 600--609, 2001.

\bibitem[KKT03]{KKT}
D.~Kempe, J.~Kleinberg, and \'{E}. Tardos.
\newblock Maximizing the spread of influence through a social network.
\newblock In {\em Proceedings of the 9th ACM SIGKDD International Conference on
  Knowledge Discovery and Data Mining}, pages 137--146, 2003.

\bibitem[KMP13]{KMP}
W.B. Kinnersley, D.~Mitsche, and P.~Pra{\l}at.
\newblock A note on the acquaintance time of random graphs.
\newblock {\em The Electronic Journal of Combinatorics}, 20(3), 2013.

\bibitem[LY94]{LY94}
C.~Lund and M.~Yannakakis.
\newblock On the hardness of approximating minimization problems.
\newblock {\em Journal of the ACM}, 41(5):960--981, 1994.

\bibitem[MP13]{MP}
T.~Muller and P.~Pra{\l}at.
\newblock The acquaintance time of (percolated) random geometric graphs.
\newblock {\em Available from {\em http://arxiv.org/abs/1312.7170}}, 2013.

\bibitem[Rei12]{Rei}
D.~Reichman.
\newblock New bounds for contagious sets.
\newblock {\em Discrete Mathematics}, 312(10):1812--1814, 2012.

\end{thebibliography}

\end{document}